\documentclass[letterpaper, 10 pt, onecolumn]{IEEEtran}

\usepackage{xparse}

\NewDocumentCommand\until{g}{%
  \mathcal{U}\IfNoValueTF{#1}{}{_{[#1]}}}
\NewDocumentCommand\globally{g}{%
  \square\IfNoValueTF{#1}{}{_{[#1]}}}
\NewDocumentCommand\eventually{g}{%
  \lozenge\IfNoValueTF{#1}{}{_{[#1]}}}
\NewDocumentCommand\release{g}{%
  \mathcal{R}\IfNoValueTF{#1}{}{_{[#1]}}}

\NewDocumentCommand\dut{g}{%
  \delta_{\mathtt{UnsafeTransition}}\IfNoValueTF{#1}{}{^{(#1)}}}
\NewDocumentCommand\dst{g}{%
  \delta_{\mathtt{SafeTransition}}\IfNoValueTF{#1}{}{^{(#1)}}}
\NewDocumentCommand\dg{g}{%
  \delta_{\mathtt{Goal}}\IfNoValueTF{#1}{}{^{(#1)}}}
\NewDocumentCommand\du{g}{%
  \delta_{\mathtt{Unsafe}}\IfNoValueTF{#1}{}{^{(#1)}}}
\NewDocumentCommand\di{g}{%
  \delta_{\mathtt{Input}}\IfNoValueTF{#1}{}{^{(#1)}}}
\NewDocumentCommand\dd{g}{%
  \delta\IfNoValueTF{#1}{}{^{(#1)}}}
\NewDocumentCommand\dSet{mg}{%
  \delta_{\mathtt{#1}}\IfNoValueTF{#2}{}{^{(#2)}}}
\NewDocumentCommand\Set{mg}{%
  \mathtt{#1}\IfNoValueTF{#2}{}{^{(#2)}}}
\NewDocumentCommand\Init{g}{%
  \mathtt{Init}\IfNoValueTF{#1}{}{^{(#1)}}}
\NewDocumentCommand\Goal{g}{%
  \mathtt{Goal}\IfNoValueTF{#1}{}{^{(#1)}}}
\NewDocumentCommand\Unsafe{g}{%
  \mathtt{Unsafe}\IfNoValueTF{#1}{}{^{(#1)}}}
\NewDocumentCommand\Input{g}{%
  \mathtt{Input}\IfNoValueTF{#1}{}{^{(#1)}}}
\NewDocumentCommand\UnsafeInput{g}{%
  \mathtt{UnsafeInput}\IfNoValueTF{#1}{}{^{(#1)}}}
\NewDocumentCommand\UnsafeTransition{g}{%
  \mathtt{UnsafeTransition}\IfNoValueTF{#1}{}{^{(#1)}}}
\NewDocumentCommand\SafeTrans{g}{%
  \mathtt{SafeTransition}\IfNoValueTF{#1}{}{^{(#1)}}}
\NewDocumentCommand\Gm{g}{%
  G_{-}\IfNoValueTF{#1}{}{^{(#1)}}}
\NewDocumentCommand\Gz{g}{%
  G_{0}\IfNoValueTF{#1}{}{^{(#1)}}}
\NewDocumentCommand\Gp{g}{%
  G_{p}\IfNoValueTF{#1}{}{^{(#1)}}}
\NewDocumentCommand\Guard{g}{%
  \mathtt{Guard}\IfNoValueTF{#1}{}{^{(#1)}}}
\NewDocumentCommand\Inv{g}{%
  \mathtt{Inv}\IfNoValueTF{#1}{}{^{(#1)}}}
\NewDocumentCommand\ta{g}{%
  t_{a}\IfNoValueTF{#1}{}{^{(#1)}}}
\NewDocumentCommand\tf{g}{%
  t_{f}\IfNoValueTF{#1}{}{^{(#1)}}}
\NewDocumentCommand\nta{g}{%
  \nominal{t}_{a}\IfNoValueTF{#1}{}{^{(#1)}}}
\NewDocumentCommand\ntf{g}{%
  \nominal{t}_{f}\IfNoValueTF{#1}{}{^{(#1)}}}

\newcommand{\nominal}[1]{\widetilde{#1}}

\usepackage{graphicx}
\usepackage{float}
\usepackage[caption=false,font=footnotesize]{subfig}
\usepackage{fixltx2e}
\usepackage{tabularx}
\usepackage{verbatim}
\usepackage{color}
\usepackage{hyperref}
\usepackage{xmpmulti}
\usepackage{transparent}

\usepackage[absolute,overlay]{textpos}
\usepackage{multirow}
\usepackage{booktabs}

\usepackage{amsmath}
\usepackage{amsfonts}
\usepackage{amssymb}

\usepackage{amsthm} % used for proof environment
\usepackage{mathtools} % allows cases environment & some formatting
\usepackage{algorithm}
\usepackage{algorithmicx} % use in-text or within a figure environment
\usepackage{algpseudocode}
\usepackage[english]{babel} % needed for \th, \st, \nd, \rd 
\usepackage{xspace}
\usepackage{enumerate}
\usepackage{etoolbox}

\newtheorem{theorem}{Theorem}

\newtheorem{definition}[theorem]{Definition}
\newtheorem{problem}[theorem]{Problem}

\newcommand{\abs}[1]{\left\lvert #1\right\rvert}
\newcommand{\norm}[1]{\left\lVert #1\right\rVert}

\newcommand{\defaultIntVar}{x}
\NewDocumentCommand\myint{O{}O{}mO{\defaultIntVar}}{%
  \int_{#1}^{#2}#3\;\mathrm{d}#4%
}
\makeatletter
\newcommand{\raisemath}[1]{\mathpalette{\raisem@th{#1}}}
\newcommand{\raisem@th}[3]{\raisebox{#1}{$#2#3$}}
\makeatother
\newcommand{\R}{\mathbb{R}}

\newcommand{\argmin}{\mathrm{arg}\min}

\providebool{isdraft}
\providebool{isoriginal}

\def\ie{\textit{i.e.}}

\title{{\Large \textbf{An MILP Approach for Real-time Optimal
      Controller Synthesis \\ with Metric Temporal Logic
      Specifications}}}

\author{Sayan Saha and A. Agung Julius\thanks{Both the authors are
    with the Department of Electrical, Computer, and Systems
    Engineering, Rensselaer Polytechnic Institute, Troy, NY 12180,
    Email: \texttt{sahas3,julia2@rpi.edu}.}}

\begin{document}

\maketitle

\begin{abstract}

  The fundamental idea of this work is to synthesize reactive
  controllers such that closed-loop execution trajectories of the
  system satisfy desired specifications that ensure correct system
  behaviors, while optimizing a desired performance criteria. In our
  approach, the correctness of a system's behavior can be defined
  according to the system's relation to the environment, for example,
  the output trajectories of the system terminate in a goal set
  without entering an unsafe set. Using Metric Temporal Logic (MTL)
  specifications we can further capture complex system behaviors and
  timing requirements, such as the output trajectories must pass
  through a number of way-points within a certain time frame before
  terminating in the goal set. Given a Mixed Logical Dynamical (MLD)
  system and system specifications in terms of MTL formula or simpler
  reach-avoid specifications, our goal is to find a closed-loop
  trajectory that satisfies the specifications, in non-deterministic
  environments. Using an MILP framework we search over the space of
  input signals to obtain such valid trajectories of the system, by
  adding constraints to satisfy the MTL formula only when necessary,
  % in a lazy iterative encoding scheme
  to avoid the exponential complexity of solving MILP problems. We also
  present experimental results for planning a path for a mobile robot
  through a dynamically changing environment with a desired task
  specification.

  
\end{abstract}

\section{Introduction}
\label{sec:intro}

At a high level, we synthesize a controller, by optimizing a desired
cost function that provides inputs to a system, such that the behavior
of the output satisfies some given requirements, such as the system
avoids some unsafe behaviors and eventually terminates at some desired
safe conditions. In recent years, a lot of attention has been given to
temporal logic constraints, which have been used extensively for
expressing reach-avoid specifications, safety requirements, and
sequencing of tasks to be performed. These allow the designer to
specify time-dependent constraints; for example, we may require that
some property will eventually hold, or that some property holds until
some other property is true. Using temporal logics allows much greater
expressivity in defining desired system behaviors than their
non-temporal counterparts, but at the cost of additional difficulty in
satisfying the constraints.

A common approach to synthesize controllers to satisfy Linear Temporal
Logic (LTL) properties is to create a finite abstraction model of the
dynamical system which can be then used to synthesize controllers
using an automata-based approach \cite{alur2000discrete,
  fainekos2009temporal, kloetzer2008fully, fainekos2005hybrid}. This
approach, however, results in high computational complexity due to
quantization of the finite abstraction model and the size of the
automaton may also be exponential in the length of the
specification. In \cite{decastro2014synthesis} the authors focus on
coarse abstractions of the state-space to alleviate the increasing
complexity problems as state-space dimension increases and also
synthesize controllers for satisfying reactive tasks. A fine
abstraction model for systems with complex dynamics is presented in
\cite{Kressgazitcorrect} to synthesize reactive controllers by
planning paths for a finite horizon, at the cost of
\textit{completeness} of the approach. Path planning using iterative
sampling-based approach, while optimizing a certain cost and
guaranteeing temporal logic specifications are presented in
\cite{Karaman2012sampling, livingston2015cross}. Recently, researchers
have been using mixed integer-linear programming to solve an optimal
control problem by encoding LTL specifications
\cite{karaman2011linear, wolff2014optimization, karaman2008optimal},
Metric Temporal Logic (MTL) specifications \cite{karaman2008vehicle},
and Signal Temporal Logic (STL) specifications \cite{raman2014model}
as mixed integer-linear constraints on the optimization
variables. Reactive controller synthesis satisfying temporal logic
specifications using receding horizon control has been considered in
\cite{wongpiromsarn2012receding, ulusoy2013receding,
  raman2015reactive}.

In this paper we consider MTL specifications, which augment the
temporal operators with a metric interval or time bounds over which
the operator is required to hold \cite{fainekos2009robustness,
  koymans1990specifying, fainekos2006robustness}. We consider the task
of determining an input signal for a Mixed Logical Dynamical (MLD)
system such that the system's output satisfies a given MTL
specification. We address this problem by casting it as a Mixed
Integer Linear Program (MILP).  This concept has been previously
applied by \cite{raman2014model, raman2015reactive} by encoding an STL
specification \(\phi\)
in terms of an MILP, where a variable is associated with each time step
and predicate, indicating the degree by which that predicate is
satisfied by the associated output trajectory point of the system. The
temporal operators and logical operators in the given specification
are then broken down into a set of boolean operators, each of which
are assigned a boolean variable along with a set of constraints that
guarantee the variable is true when the boolean operator is satisfied,
and the boolean variable is false when the operator is not satisfied.
In all, this formulation introduces $O(N\cdot\abs{\phi})$ boolean
variables and constraints, where $\abs{\phi}$ denotes the number of
operators in the STL specification and \(N\)
is the horizon length over which the input signal is to be
determined. We expect that linearly increasing the length of the
trajectory or length of the STL specification will cause the
time-complexity to grow exponentially since solving a MILP is
exponential in the number of binary decision variables in the
worst-case scenario \cite{garey1979computers}. This can quickly render
the MILP intractable even for the relatively small problems. This
paper is motivated to circumvent this issue, such that the controller
synthesis problem can be solved in real-time and can be applied in
practical applications of robot motion path planning with temporal
specifications. We propose a method whereby we dynamically identify
the critical time-points over the simulation horizon, where the system
trajectory violates the given MTL specifications most and introduce
boolean variables and constraints only for those critical time-points
iteratively and resolve the problem till the system trajectory
satisfies the specification. This leads to a much smaller MILP problem
to be solved leading to a significant reduction in the time required
to generate the input signal. This approach is very closely related to
the formulation presented in \cite{earl2005iterative} for generating
trajectories to avoid obstacles. We show the effectiveness of our
algorithm by running experiments on a \textit{m3pi} robot to plan a
trajectory through a dynamically changing environment in real-time, so
as to not hit any obstacles and reach a desired location within a
given time limit.

\section{Preliminaries}
\label{sec:preliminaries}
We consider discrete-time systems of the form
\begin{align}
  \label{eq:sys-dyn}
  x_{k+1} = F(x_{k}, u_{k}),
\end{align}
where,
\(x_{k} \in \mathcal{X} \subseteq \mathbb{R}^{n_{c}} \times
\{0,1\}^{n_{l}}\)
are the continuous and binary/logical states, and
\(u_{k} \in U \subseteq \mathbb{R}^{m_{c}} \times \{0,1\}^{m_{l}}\)
are the continuous and binary/logical control inputs at the time
indices \(k = 0, 1, \ldots \).
We denote the system trajectory at time index \(k\)
under the control input
\(\mathbf{u}^{k} = \{u_{1}, u_{2}, \cdots, u_{k}\}\)
starting from a given initial condition \(x_{0} \in \mathcal{X}\),
by \(x(\mathbf{u}^{k}) = \{x_{0}, x_{1}, \cdots, x_{k}\}\).
This system model provides a set of constraints for the optimization
procedure, such that at any time index \(k\)
the resulting trajectory \(x(\mathbf{u}^{k})\)
satisfies the system dynamics given in (\ref{eq:sys-dyn}). These
constraints can be easily formulated in terms of mixed integer-linear
program if the system under consideration belongs to the class of
mixed-logical dynamical systems \cite{Bemporad99}, that includes
linear hybrid systems, constrained or unconstrained linear systems,
piece-wise affine systems. Differentially flat and feedback
linearizable systems \cite{Khalil, Julius12} can also be
considered if the temporal logic specifications are in terms of the
flat and observable outputs respectively.

An MTL formula is a formal language, that can be used to express
desired properties that a system must satisfy with certain timing
requirements. We consider the temporal
operators \textit{eventually} (\( \eventually{\mathcal{I}} \)),
\textit{always} (\( \globally{\mathcal{I}} \))
and \textit{until} (\( \until{\mathcal{I}}\)),
and logical operators, such as, \textit{conjunction} (\(\wedge\)),
\textit{disjunction} (\(\vee\)),
\textit{negation} (\(\neg\)),
and \textit{implication} ($\rightarrow$), that can be used to combine
\textit{atomic propositions} to form the MTL formula. We associate a
set \(\mathcal{O}(p)\subseteq \mathcal{X}\)
with each atomic proposition \(p\),
such that \(p\)
is true at time index $k$ if and only if $x_{k}\in\mathcal{O}(p)$.

For example, using MTL one can easily express a desired system
behavior that ``the system trajectory should never enter some unsafe
set \(\mathcal{O}(p_{\Unsafe})\) and terminate in some desired set
\(\mathcal{O}(p_{\Goal})\) within time \(t_{3}\) to \(t_{4}\) and
should pass through the set \(\mathcal{O}(p_{\mathcal{W}})\) within
the time frame \(t_{1}\) to \(t_{2}\) and must stay in
\(\mathcal{O}(p_{\mathcal{W}})\) once the trajectory enters it for
\(s\) units of time'' as
\begin{align*}
  \phi = \globally \neg p_{\Unsafe} \wedge \eventually_{[t_{3}, t_{4}]}
  p_{\Goal} \wedge \eventually_{[t_{1}, t_{2}]} \globally_{[0,s]}
  p_{\mathcal{W}}.
\end{align*}

\begin{definition}
\label{def:d1} We define a system trajectory that satisfies the
dynamics given in (\ref{eq:sys-dyn}) to be an (in)feasible trajectory
if it (falsifies)satisfies the given MTL specification.
\end{definition}

The robustness measure, \(\rho_{\phi}\)
of a system trajectory defines how robustly a system trajectory
satisfies or falsifies the given MTL specification. The measure takes
positive values if the trajectory satisfies the specification, and
negative values otherwise. Intuitively, the robustness degree of an
(in)feasible trajectory \(x(\mathbf{u}^{k})\)
is the largest distance that we can independently perturb the points
along the trajectory and maintain (in)feasibility. This defines a tube
around the original trajectory such that any trajectory within this
tube is guaranteed to satisfy (or falsify) the specification.

This concept is demonstrated in Fig.~\ref{fig:robBall}.  In this
figure we consider the MTL specification
\begin{align*}
  \phi^{\prime} &= \globally{t_{0}, t_{2}} \neg p_{\mathcal{B}} \wedge
  \globally{T,T}
  p_{\mathcal{A}},
\end{align*}
which states that between times $t_{0}$ and $t_{2}$ the trajectory
should avoid the set $\mathcal{B}$ and at time $T$ the trajectory is
in set $\mathcal{A}$.  Two trajectories are shown:
\begin{description}
\item[\( x(\mathbf{u}^{k}_{1})\)] ~ an infeasible trajectory;
\item[\( x(\mathbf{u}^{k}_{2})\)] ~ a feasible trajectory.
\end{description}
Note that the trajectory \( x(\mathbf{u}^{k}_{1})\)
has a large negative robustness, denoted by $\rho_{\phi^{\prime}}^{1}$
and determined by the distance between the trajectory at time $t_{0}$,
\ie, \(x_{k_{t_{0}}}\)
(\(k_{t}\)
represents the time index corresponding to the time \(t\)),
and the boundary of the predicate $\mathcal{B}$.  We call this time
\(t_{0}\)
the \emph{critical time} and the predicate \(\mathcal{B}\)
the \emph{critical predicate}.  By moving the \emph{critical point}
\(x_{k_{t_{0}}}\)
on the trajectory by any amount greater than this distance, we can
push it outside of $\mathcal{B}$ in order to satisfy the
specification.  Analogously, trajectory \( x(\mathbf{u}^{k}_{2})\)
has a smaller positive robustness, denoted by
$\rho_{\phi^{\prime}}^{2}$ and determined by the distance between the
trajectory at time $t_{1}$ and the boundary of the predicate
$\mathcal{B}$.  This is because the trajectory at $t_{1}$ is closest
to falsifying the specification, and would need to be moved by at
least this amount in order to do so.

The tool \texttt{TaLiRo} \cite{annpureddy2011s} can be used to
calculate the robustness, critical time, and critical predicate for a
trajectory and MTL specification.  The latter two values will be used
to determine which constraints to add at each iteration of our method,
presented in Section \ref{sec:ctrlSynthesis}.

Finally, we assume that the set $\mathcal{O}(p)$ associated with each
predicate $p$ is polyhedral, defined by \(f\)
faces.  In this case, we can represent the set of points in the set
uniquely by
$\mathcal{O}(p) = \{x~|~Ax \leq b,~A\in\R^{f\times n}, ~ n = n_{c} +
n_{l}, ~b\in\R^{f}\}$,
with rows of \(A\)
normalized to unit vectors. Further, we assume the MTL formula is
expressed in negation normal form, which requires that the formula be
transformed such that all negations, $\neg$, only appear immediately
before a predicate. If we require predicate $p$ to be true at time
index \(k\),
that is, the trajectory at time index \(k\)
has to be in $\mathcal{O}(p)$, then we can represent this via the
constraints
\begin{align}
  \label{eq:1}
  Ax_{k} \leq b.
\end{align}
This requires the trajectory point \(x_{k}\)
to lie in the intersection of all of the half-spaces of the faces of
the predicate set \(\mathcal{O}(p)\),
which by definition is polyhedron.

Conversely, if we require predicate $p$ to be false at time index $k$,
that is, the trajectory at time $k$ is not in $\mathcal{O}(p)$, then
we can represent this via the constraints using the \textit{big-M}
formulation method by introducing a new binary decision vector \(z\)
as
\begin{align}
  \label{eq:3}
  Ax_{k} + Mz & \geq b, M\in\R_{+},~z\in \{0,1\}^{f}, \\
  \label{eq:4}
  \sum_{i} z_{i} &\leq f - 1,
\end{align}
where $M$ is a value large enough to make the corresponding constraint
hold for any allowable value of $x_{k}$. This requires the point
\(x_{k}\)
to lie outside of at least one half-space defining the polyhedron,
which is both necessary and sufficient for the point \(x_{k}\)
to lie outside of the polyhedron.

\begin{figure}[tp]
  \begin{center}
    \includegraphics[scale = 0.65]{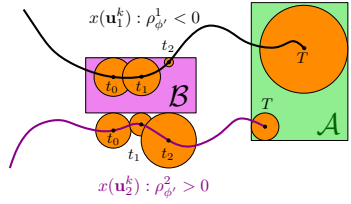}
    \caption{Illustration of robustness degree for MTL specification.}
    \label{fig:robBall}
  \end{center}
  \vspace{-0.25in}
\end{figure}

\section{Controller Synthesis}
\label{sec:ctrlSynthesis}
\begin{problem}
  \label{prb:1}
  Given an MLD system of the form (\ref{eq:sys-dyn}), initial state
  \(x_{0} \in \mathcal{X}\),
  trajectory length \(N\),
  correct system behavior defined in terms of an MTL specification
  \(\phi\),
  a desired robustness measure \(\rho_{\phi}^{d}\)
  and a performance objective \(J\), find
  \begin{align*}
    \underset{\mathbf{u}^{N}}{\argmin} ~ & J(x(\mathbf{u}^{N}), \mathbf{u}^{N}) \\
    \text{subject to} ~ & \rho_{\phi}(x(\mathbf{u}^{N})) \geq \rho_{\phi}^{d}.
  \end{align*}
\end{problem}

\begin{algorithm}[tp]
  \caption{Open Loop Controller Synthesis}
  \label{alg:1}
  \begin{algorithmic}[1]
   \State Initialize MILP with continuous state and input variables
    and constraints enforcing the linear dynamics in
    (\ref{eq:sys-dyn})
    \begin{align*}
      \text{MILP-cur} := \left \{ \begin{aligned}
          \underset{\mathbf{u}^{N}}{\argmin} ~ 
          & J(x(\mathbf{u}^{N}), \mathbf{u}^{N})\\
          \text{s.t.} ~ x_{k+1} = & F(x_{k}, u_{k}), ~ x(\mathbf{u}^{N})
          = \{x_{i}\}_{i = 0}^{N}
        \end{aligned} \right.
    \end{align*}
    \State Solve MILP-cur
    \State Run \texttt{TaLiRo} on resulting trajectory to determine
    robustness, critical time, and critical predicate.
    \While{robustness $< 0$}
      \If{critical predicate is a safe predicate} 
      \State \begin{align*}
               \text{MILP-cur} := \left \{\begin{aligned}
                                     & \text{MILP-cur}, \\
                                     \text{s.t.} ~ & Ax_{k} \leq b.
                                   \end{aligned} \right .
             \end{align*}
      \Else
      \State \begin{align*}
               \text{MILP-cur} := \left \{ \begin{aligned}
                                          & \text{MILP-cur} \\
                                          \text{s.t.} ~ & Ax_{k} + Mz \geq
                                          b,\\
                                          & \sum_{i} z_{i} \leq f - 1.
                                        \end{aligned} \right .
             \end{align*}
             \EndIf
             \State Resolve MILP-cur
             \If{MILP-cur is an infeasible problem}
             \State Return ``No feasible trajectory exists.''
           \Else   
           \State Rerun \texttt{TaLiRo}
           \EndIf
           \EndWhile
           \State Return \(x(\mathbf{u}^{N}), \mathbf{u}^{N}\)
    \end{algorithmic}
\end{algorithm}
 
We propose a heuristic approach given in Algorithm~\ref{alg:1} to
solve the controller synthesis problem. The basic idea is to run the
MILP first with only the system dynamic constraints (i.e. without MTL
constraints). This formulation only contains boolean variables
associated with the system dynamics in (\ref{eq:sys-dyn}). If the MTL
constraints are already satisfied, then we are done. Otherwise, we
find the point that in some sense corresponds to the largest violation
of the specification, and require this point to satisfy the
corresponding predicate, and repeat.  If the point is required to lie
inside the polyhedron, then only the linear constraints \eqref{eq:1}
need to be added.  If the point is required to lie outside of the
polyhedron, then the $ f $ binary variables $z$ and the linear
constraints \eqref{eq:3} need to be added. However, note that the
solution of the MILP problem can satisfy the constraints given by
(\ref{eq:1}) or (\ref{eq:3}) at equality and hence the resulting
system trajectory though a feasible one, will have robustness measure
\(\rho_{\phi}(x(\mathbf{u}^{N})) \geq 0\).
Hence, in order to obtain a system trajectory with a desired
robustness measure \(\rho_{\phi}^{d}\),
we resize the predicate sets \(\mathcal{O}(p)\)
by \(\rho_{\phi}^{d}\).
If the predicate \(p\)
is such that it is a safe (unsafe) predicate, or in other words, the
system can (never) visit the predicate, we shrink (bloat) the size
of the predicate such that the new predicate set is given by
\begin{align}
  \tilde{\mathcal{O}}(p) = \{x ~|~ Ax \leq \tilde{b}\}, ~ \tilde{b} =
  b \pm \rho_{\phi}^{d} \mathbf{1}_{f},
  \label{eq:predRobust}
\end{align}
where, \(\mathbf{1}_{f}\) is a column vector of ones of size \(f\).
Identifying the safe and unsafe predicates can be done easily by
parsing through the given MTL specification and marking the predicates
with a negation (\(\neg\)) in front of them to be unsafe and the rest
to be safe predicates.

\begin{figure*}[!t]
  \centering \subfloat[\label{subfig:a} Initial infeasible trajectory
     \(x(\mathbf{u}^{k}_{1})\):
     critical constraint corresponds to
     \(t_{0}\).]{\includegraphics[scale=0.65]{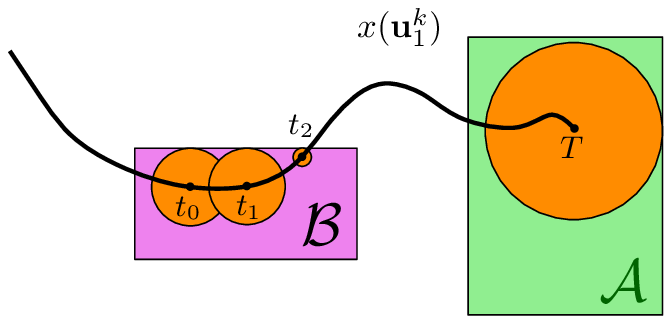}}
       \hfil \subfloat[\label{subfig:b} Infeasible solution trajectory
     \(x(\mathbf{u}^{k}_{1^{\prime}})\)
     after solving the MILP first time: critical constraint corresponds
     to
     \(t_{1}\).]{\includegraphics[scale=0.65]{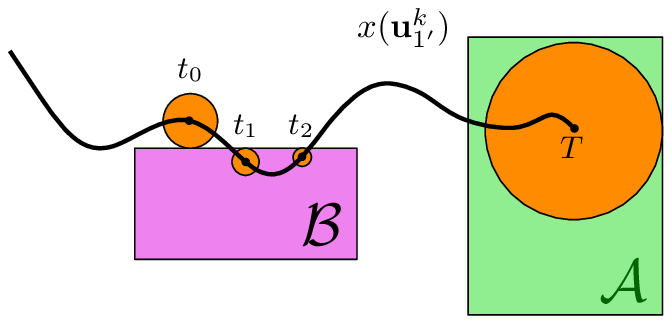}}
     \hfil \subfloat[\label{subfig:c} Feasible solution trajectory
     \(x(\mathbf{u}^{k}_{1^{\prime\prime}})\)
     after solving the MILP second
     time.]{\includegraphics[scale=0.65]{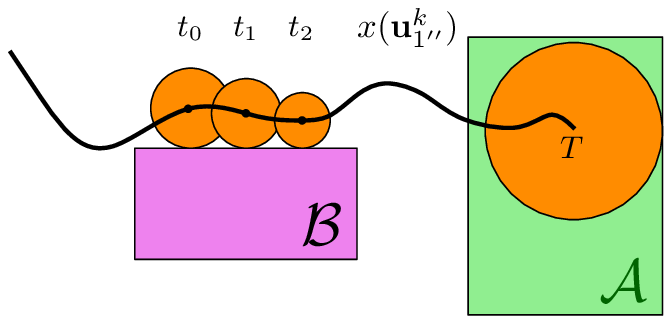}}
    \caption{Illustration of iterative addition of critical constraints
     in MILP formulation.} 
   \label{fig:critConsIter}
    \vspace{-0.2in}
\end{figure*}

The motivating idea behind this approach is that we do not require to
put the constraints given by (\ref{eq:1})-(\ref{eq:3}) at every time
points of the trajectory; constraining only a few critical time points
will result in the whole system trajectory to satisfy the MTL
specifications in most cases. For instance (see Fig.
\ref{fig:critConsIter}), if we want the infeasible trajectory
\(x(\mathbf{u}^{k}_{1})\)
to satisfy the MTL specification \(\phi^{\prime}\)
we first identify that the critical time and predicate are the time
\(t_{0}\)
and \(\mathcal{B}\)
respectively by running \texttt{TaLiRo} as in Step \(3\)
of Algorithm \ref{alg:1}. The critical constraint corresponding to the
time \(t_{0}\)
is then added to the MILP formulation, such that the trajectory point
at \(t_{0}\)
is now constrained to lie outside the critical predicate
\(\mathcal{B}\)
and the new MILP is solved as in steps \(5\)
and \(6\)
of Algorithm \ref{alg:1}. Assume, the solution of the MILP is
\(x(\mathbf{u}^{k}_{1^{\prime}})\).
Since this resulting trajectory is still infeasible as found in Step
\(7\),
we repeat the process once again by adding constraints corresponding
to the time \(t_{1}\).
After resolving the newer MILP, assume we obtain the trajectory
\(x(\mathbf{u}^{k}_{1^{\prime\prime}})\),
which turned out to be a feasible trajectory and hence the Algorithm
\ref{alg:1} terminates. Note that we only added constraints for two
time instances and do not need to introduce the additional binary
variables associated with the trajectory point at time \(t_{2}\).
The goal is to reduce the computation time by iteratively solving
much smaller MILPs, rather than the one large MILP.

However, iteratively adding critical constraints to the MILP problem
in this way will result in an additional overhead in terms of encoding
the MILP problem each time a new constraint is added. To circumvent this
issue, we introduce a new scalar binary parameter \(\sigma_{(i,j)}\)
for each predicate \(i\)
and each time index \(j\)
and modify the linear constraints given in (\ref{eq:1}) and
(\ref{eq:3}) for safe and unsafe predicates respectively as
\begin{align}
  \begin{aligned}
    Ax_{k}\sigma_{(\cdot, k)} \leq \tilde{b} \sigma_{(\cdot, k)}, \\
    (Ax_{k} + Mz)\sigma_{(\cdot, k)} \geq \tilde{b} \sigma_{(\cdot, k)}.
  \end{aligned}
  \label{eq:4}
\end{align}
Note that if \(\sigma_{(\cdot, k)} = 0\),
then the constraint is trivially satisfied and hence relaxed. Only
when \(\sigma_{(\cdot, k)} = 1\)
the constraint is required to hold. Using this modification, we encode
the linear constraints for all the predicates at all the time indices
\( k = \{0, 1, \cdots, N\}\)
based on whether the predicate is safe and unsafe right at the
beginning. We also set all \(\{\sigma_{(i,j)}\}\) to zero. Afterwards,
as we proceed through the steps of Algorithm \ref{alg:1}, we set one
of \(\{\sigma_{(i,j)}\}\) to \(1\) corresponding to the critical
predicate and the critical time for that iteration. In this way, we
add the new constraints to the MILP problem iteratively without having
to encode the problem at each step.

In general, this method is not guaranteed to
find a feasible solution if one exists, and is not guaranteed to
return the optimal trajectory, even if a feasible solution is
found. However for a fragment of MTL specifications we have the
following result.
\begin{theorem}
  \label{thm:1}
  If the MTL specification consists of only conjunctions ($\wedge$)
  and the globally operator $(\globally)$, then if Algorithm
  \ref{alg:1} finds a feasible trajectory then it is the
  optimal trajectory. If Algorithm \ref{alg:1} fails to find any
  feasible trajectory, then Problem \ref{prb:1} is infeasible.
\end{theorem}

\begin{proof}
  A full-scale optimization approach for solving Problem \ref{prb:1}
  with such MTL specifications will require constraints given in
  either (\ref{eq:1}) or (\ref{eq:3}) to hold at all the time indices
  for all the predicates in the specification. Whereas, Algorithm
  \ref{alg:1} requires only a subset of those constraints to hold. Let
  us denote the search-spaces explored in full-scale optimization and
  in the final iteration of Algorithm \ref{alg:1} by \(U^{*}\)
  and \(U^{\prime}\)
  respectively. Since, any feasible trajectory for the full-scale
  optimization problem is also a feasible trajectory for Algorithm
  \ref{alg:1} we have that, \( U^{*} \subseteq U^{\prime}\).
  Now, if Algorithm \ref{alg:1} returns a feasible solution optimal
  over \(U^{\prime}\),
  it should also be a feasible solution in the \(U^{*}\)
  space, according to construction of Algorithm \ref{alg:1}, implying
  it is also the optimal solution. However, if the
  MTL specification involves \textit{disjunction} or other temporal
  operators that can be broken down in terms of \textit{disjunction}
  (either this holds \textit{or} that holds), then not all feasible
  trajectories of full-scale optimization problem are feasible for
  Algorithm \ref{alg:1}. If some disjunction is the critical portion
  of the specification that defines the robustness, Algorithm
  \ref{alg:1} will require the trajectory to satisfy one particular
  predicate in the \textit{disjunction}, however the full-scale
  optimization will require the trajectory to satisfy any of the
  predicates in the \textit{disjunction} and hence
  \(U^{*} \not \subseteq U^{'}\).
\end{proof}

For our numerical examples, presented in Section \ref{sec:examples} we
consider the performance criteria
\begin{align}
  \label{eq:objfunc}
  J = \sum\limits_{k=1}^{N} \norm{\text{diag}(\mathbf{R}) u_{k}}_{1},
\end{align}
to minimize the control effort, where \(\text{diag}(\mathbf{R})\)
is a diagonal matrix consisting of the elements of the non-negative
weighting row-vector \(\mathbf{R}\).
Introducing slack variables $s_{kj}$, with $k = 1, \ldots, N$,
$j = 1, \ldots, m$, \(m = m_{c} + m_{l}\),
the objective function can be reformulated as
\begin{align}
  \label{eq:reform-obj}
  J = \sum\limits_{k=1}^{N} \mathbf{R} s_{k},
\end{align}
with an added set of constraints,
\begin{align}
  \label{eq:obj-constraints}
   -s_{kj} \leq u_{kj} \leq s_{kj}
\end{align}
where, \(u_{kj}\) denotes the \(j\)-th component of the control input
\(u_{k}\) at time index \(k\), resulting in a linear optimization
problem.

The trajectory length \(N\) depends on the \textit{bound} of a
\textit{bounded-time} temporal specification (does not contain any
unbounded operator), which is computed to be the ``maximum over the sums
of all nested upper bounds on temporal operators''
\cite{raman2014model}.

\section{Reactive Controller Synthesis}
\label{sec:react-contr-synth}

In this section, we present a receding horizon controller (RHC)
framework to make the MILP controller presented in Section
\ref{sec:ctrlSynthesis}, reactive to dynamically changing predicate
sets. At each step \(i = \{1, 2, \cdots, N\}\)
of the RHC computation, we keep track of the system trajectory for
time indices \( \{0, 1 , \cdots, (i-1)\}\)
and search for a system trajectory of length \((N - i + 1)\),
starting from the initial condition \(x_{i-1}\),
such that the combined system trajectory satisfies the given MTL
specification with the desired robustness measure. While computing for
the system trajectory at step \(6\)
of the Algorithm \ref{alg:2}, we use the open loop controller
presented in the Algorithm \ref{alg:1}, with one slight modification:
searching for the critical time and critical predicate is based on the
combined system trajectory
\(\{x_{0}, \cdots, x_{i-2}, x_{i-1}, \hat{x}_{i}, \cdots,
\hat{x}_{N-i+1}\}\).
Also since, each predicate set can be defined uniquely in terms of the
\((A, b)\)
pair as given in (\ref{eq:predRobust}), the MILP controller takes the
\((A,b)\)
pair as arguments, instead of constant matrices, to account for any
predicate sets that might change through the execution of the whole
control plan.

As the end goal of this work is real-time controller
synthesis, ideally it is desired that the steps \(3\)
to \(8\)
of the Algorithm \ref{alg:2} be executed within the time-step of the
system dynamics, such that at each time-step the system has a control
input signal to execute. However, this is usually not the case in the
experiments we have performed. Step \(6\)
is the rate limiting step of the process, since it involves
iteratively solving an MILP problem if the predicate sets are
dynamically changing and the upper bound to the number of iterations
required to find a feasible system trajectory is
\(O(N \cdot \vert p \vert )\),
where \(\vert p \vert\)
is the number of predicate sets in the MTL specification, in the
worst-case scenario. But the incremental nature of Algorithm
\ref{alg:1} allows us to execute the steps \(3\)
to \(8\)
only for the duration of the time-step of the system
dynamics. Consider the example presented in Fig. \ref{fig:react},
where instead of waiting for a feasible trajectory to be found in a
single iteration (time for which may exceed the time-step), the
controller iteratively finds a system trajectory that approaches
towards a feasible one and requires \(3\)
iterations to find a feasible trajectory after a new unsafe region was
introduced in the workspace. Thus, obtaining a feasible trajectory
with the desired robustness may take a few iterations, but the
Algorithm \ref{alg:2} can essentially be run in real-time for this
example.

\begin{algorithm}[tp]
  \caption{Receding Horizon Controller Synthesis}
  \label{alg:2}
  \begin{algorithmic}[1]
    \State Run Algorithm \ref{alg:1} for trajectory length \(N\).
    \State Set initial condition to be: \(x_{0}\).
    \For{\(i = \{1, \cdots, N+1\}\)}
    \State Set past system trajectory: \( \{x_{0}, \cdots,
    x_{i-1}\} \).
    \State Compute the \((A, b)\) pair defining the predicate sets.
    \State Obtain new system trajectory by running Algorithm
    \ref{alg:1} 
    for trajectory length \((N - i + 1)\): \(\{x_{i-1}, \hat{x}_{i}
    \cdots, \hat{x}_{N-i+1}\}\).
    \State Set initial condition for the next iteration to be: \(\hat{x}_{i}\).
    \EndFor    
 \end{algorithmic}
\end{algorithm}
  
\section{Examples}
\label{sec:examples}

\subsection{Numerical Example}
\label{sec:numerical-example}

For numerical simulations, we consider two different MTL
specifications for motion planning of a mobile robot with unicycle
dynamics, and compare our results by running the same examples with
the \texttt{BluSTL}\footnote{https://github.com/BluSTL/BluSTL/}
toolbox developed using the ideas in \cite{raman2014model,
  raman2015reactive}. All the simulations were performed on a computer
with a \(3.4\)
GHz \textit{Intel core i7} processor with \(16\)
GB of memory using \texttt{Gurobi}\footnote{http://www.gurobi.com/}
solver through \texttt{YALMIP} \cite{lofberg2004yalmip}.

We first consider a simple reach-avoid scenario (see Fig.
\ref{fig:m3pi-simple-MTL}), where the system has to avoid an unsafe
area in its workspace at all times and reach the goal area and stay
there from \(8.5\)
to \(10\)
seconds. These requirements can be represented using MTL
specifications as,
\vspace{-0.025in}
\begin{align*}
  \phi_{1} = \left( \globally \neg p_{\Unsafe} \right) \wedge \left(
  \globally_{[8.5, 10]} ~ p_{\Goal} \right).
\end{align*}
For the second example we consider both the \textit{eventually} and
\textit{globally} operator in a nested fashion to show that our
approach can handle complex task specifications as well. In this
scenario the requirement is still to avoid the unsafe region always
and to eventually reach the goal region sometime within \(5.5\)
to \(7.5\)
seconds and stay there for \(1.5\)
seconds. Formally, the specification is
\begin{align*}
  \phi_{2} = \left(\globally \neg p_{\Unsafe} \right) ~ \wedge
  \eventually_{[5.5, 7.5]} ~
  \left( \globally_{[0,1.5]} ~ p_{\Goal} \right).
\end{align*}
In order to implement our approach we first feedback linearize (see
Section \ref{sec:pract-example}) the unicycle dynamics to obtain
double integrator dynamics in \(X\)
and \(Y\)
directions, governing the evolution of the position of the mobile
robot. We then discretize the continuous system dynamics with a \(0.5\)
second sample time, resulting in a trajectory length of \(N = 20\)
for the cases considered here. 

\begin{figure}[H]
  \centering
  \includegraphics[width=21pc]{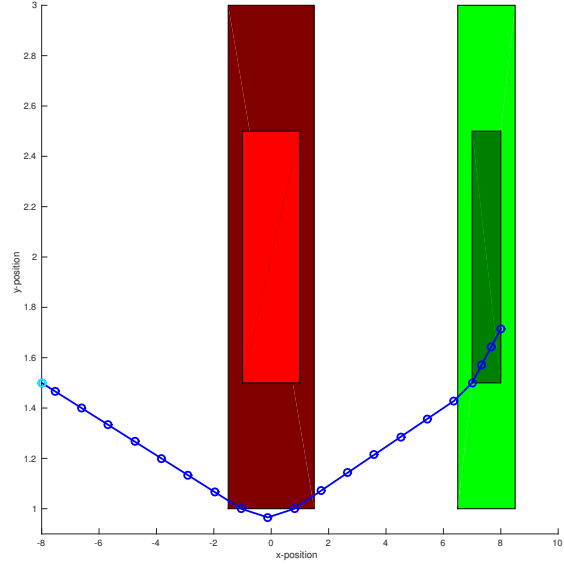}
  % \vspace{-0.5in}
  \caption{Illustration of path planning of mobile robot with the
    simple reach-avoid specification. The unsafe and the goal regions
    are colored red and green respectively. The unsafe set is bloated
    and the goal set is shrunk by the desired robustness degree,
    chosen to be \(0.5\)
    (trajectory points should not be inside the bloated unsafe set and
    outside of the shrunk goal set at the relevant time-instants). The
    system starts from the cyan colored position on the left side of
    the figure. Both BluSTL and our approach produce the same solution
    trajectory.}
  \label{fig:m3pi-simple-MTL}
\end{figure}

  \begin{table}[H]
    \centering
    \begin{center}
    % \resizebox{21pc}{!}{
      \begin{tabular}{c c c c c}
Method & MTL spec & YALMIP Time (s)  & Open Loop Time (s) &
  RHC Time (s)\\
\toprule
Our Approach & \(\phi_{1}\)   & 3.12 \(\pm\) 0.49 & 0.98 \(\pm\) 0.13 & 10.02 \(\pm\)
  0.004 \\
& \(\phi_{2}\)   & 3.02 \(\pm\) 0.78 & 0.84 \(\pm\) 0.04 & 10.02
                                           \(\pm\)
                                                           0.001 \\
        \midrule
\texttt{BluSTL} & \(\phi_{1}\)   & 38.22 \(\pm\) 2.16 & 5.83 \(\pm\) 0.21 & 60.62 \(\pm\) 1.06 \\
& \(\phi_{2}\)   & 39.21 \(\pm\) 2.52 & 77.63 \(\pm\) 2.15 & 1401.37
                                                             \(\pm\)
                                                             19.50 \\ 
% \(\phi_{1, \text{quadrotor}}\)   & 6 & 2 & 1 \\
% \(\phi_{2, \text{quadrotor}}\)   & 6 & 2 & 1 \\
        \bottomrule 
      \end{tabular}
     % }
    \vspace{0.05in}
\caption{Time taken to solve the path planning problem.}
\label{tab:timing-results}
\end{center}
\end{table}

\begin{figure}[H]
%  \vspace{-0.5in}
  \centering
  \includegraphics[width=21pc]{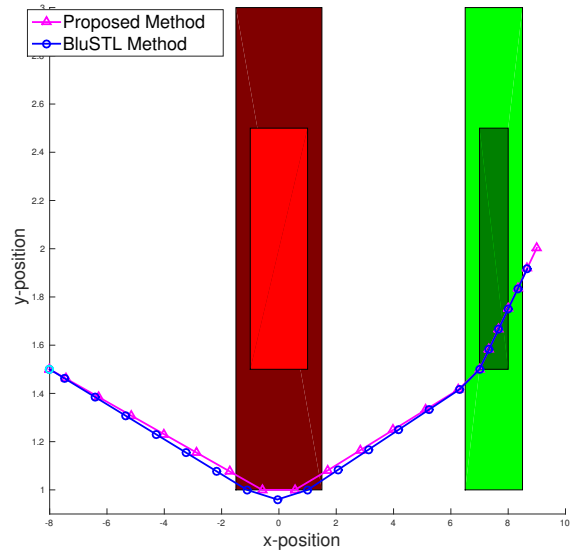}
%  \vspace{-0.4in}
  \caption{Path planning solutions of the mobile robot with complex
    task specifications. The workspace is same as before.}
  \label{fig:m3pi-ev-MTL}
\end{figure}

As expected, based on Theorem \ref{thm:1}, both \texttt{BluSTL} and
our approach produce the same optimal path for the specification
\(\phi_{1}\)
as shown in Fig. \ref{fig:m3pi-simple-MTL}. However, for the task
specification \(\phi_{2}\),
because of the presence of the \textit{eventually} operator in the
specification \texttt{BluSTL} provides a more optimal path than our
approach\footnote{Open loop and closed loop receding horizon
  implementation in the BluSTL toolbox actually produced different
  trajectories even with the exact same system parameters and
  specifications; closed loop trajectory obtained was non-optimal.},
even though both the solutions satisfy the specification \(\phi_{2}\)
with the desired robustness as shown in Fig. \ref{fig:m3pi-ev-MTL}. In
Table \ref{tab:timing-results}, we present the results obtained for
the path planning problems in 'mean \(\pm\)
standard deviation' format obtained over \(10\)
independent runs of the same problem. \textit{YALMIP Time} represents
the time taken to encode the controller and \textit{Open Loop Time}
and \textit{RHC Time} are the times required to generate the feasible
path in the open loop fashion and by using the receding horizon
controller approach respectively using the \texttt{Gurobi}
solver. Note that, the \textit{RHC Time} represents the time taken for
planning the path over \(N = 20\)
time-steps each of which is of duration \(0.5\)
seconds. The timing results clearly shows that our approach is much
faster in obtaining the solution trajectory as compared to
\texttt{BluSTL}, specifically in the case of the more complex
specification \(\phi_{2}\).
As our end goal is to implement this controller synthesis procedure in
a practical situation, being able to plan the path in a short amount
of time is of a great importance.

\subsection{Practical Example}
\label{sec:pract-example}

Experimentally, we determine the efficacy of our MILP approach for
reactive controller synthesis using a \textit{m3pi} robot with a
differential drive system dynamics. The \textit{m3pi} robot consists
of a \textit{3pi} robot base connected to a \textit{m3pi} expansion
board that allows us to communicate with
the robot using \textit{XBee} wireless communication module to send
control input signals from a workstation to the robot.

% \begin{figure}[h]
%   \centering
%   \includegraphics[width=21pc]{./pics/m3piPic.jpg}
%   \caption{Front view of the \textit{m3pi} robot}
%   \label{fig:m3pi}
% \end{figure}

Denoting the
position of the robot in a \(2\)D
plane to be \((x, y)\),
the equations of motion governing the system dynamics are
\begin{align}
  \begin{aligned}
  \dot{x} &= \frac{v_{r} + v_{l}}{2} \cos(\theta) = v \cos(\theta), \\
  \dot{y} &= \frac{v_{r} + v_{l}}{2} \sin(\theta) = v \sin(\theta),\\
  \dot{\theta} &= \frac{v_{r} - v_{l}}{2d} = \omega,
  \label{eq:sys-dyn-unicycle}
\end{aligned}
\end{align}
where, $\theta$ denotes the orientation of the robot with respect to
the coordinate frame of reference and $v_{r}$ and $v_{l}$ are the
wheel speeds of the right and left wheels respectively and are the
control input signals to the robot. Linear and angular velocities of
the robot are denoted by \(v\) and \(\omega\) respectively for ease of
notation and \(d\) is the distance of any one wheel from the center of
the robot base. To design a controller for the
unicycle agents we utilize a well-known theory from nonlinear control
called feedback linearization \cite{Khalil} so that the nonlinear
dynamics presented in (\ref{eq:sys-dyn-unicycle}) can be represented
by second order particle dynamics in two dimensions. Choosing the
position of the robot as the system output notice that,
%\vspace{-0.3in}
\begin{align*}
  \begin{aligned}
   \begin{bmatrix}
      \ddot{x} \\ \ddot{y} 
    \end{bmatrix} =
    \begin{bmatrix}
      \cos(\theta) & -\sin(\theta) \\ \sin(\theta) &
      \cos(\theta)
    \end{bmatrix}
    \begin{bmatrix}
      \dot{v} \\ v \omega
    \end{bmatrix}.
  \end{aligned}
 % \vspace{-0.1in}
\end{align*}

Then, by choosing the intermediate control inputs to the robot to be,
\(\dot{v}\) and \(v\omega\) such that,
%\vspace{-0.2in}
\begin{align}
  \begin{bmatrix}
    \dot{v} \\ v \omega
  \end{bmatrix} = \begin{bmatrix} \cos(\theta) & \sin(\theta) \\
    -\sin(\theta) & \cos(\theta)
  \end{bmatrix}
  \begin{bmatrix}
    u_{x} \\ u_{y}
  \end{bmatrix},
  \label{eq:non-linear-ctrl}
\end{align} where, \(u_{x}\) and \(u_{y}\) are the new control
inputs to be determined, leads to the input-output feedback linearized
system corresponding to second order particle dynamics in the form of
a chain of integrators as,
\begin{align}
  \begin{bmatrix}
    \ddot{x} \\ \ddot{y}
  \end{bmatrix} =
  \begin{bmatrix}
    u_{x} \\ u_{y}
  \end{bmatrix}.
\end{align}
These new control inputs \(u_{x}\)
and \(u_{y}\)
are then generated using the MILP approach presented in Section
\ref{sec:ctrlSynthesis}. Using (\ref{eq:non-linear-ctrl}), one can
derive the control inputs \(v\)
and \(\omega\), and then the original control inputs
\begin{align*}
  v_{r} = v + d\omega \quad \text{and} \quad  v_{l} = v - d\omega,
\end{align*}
so that the resulting system trajectory of the system satisfies the
given MTL specification.

\begin{figure}[H]
  \centering
  \includegraphics[scale=0.3]{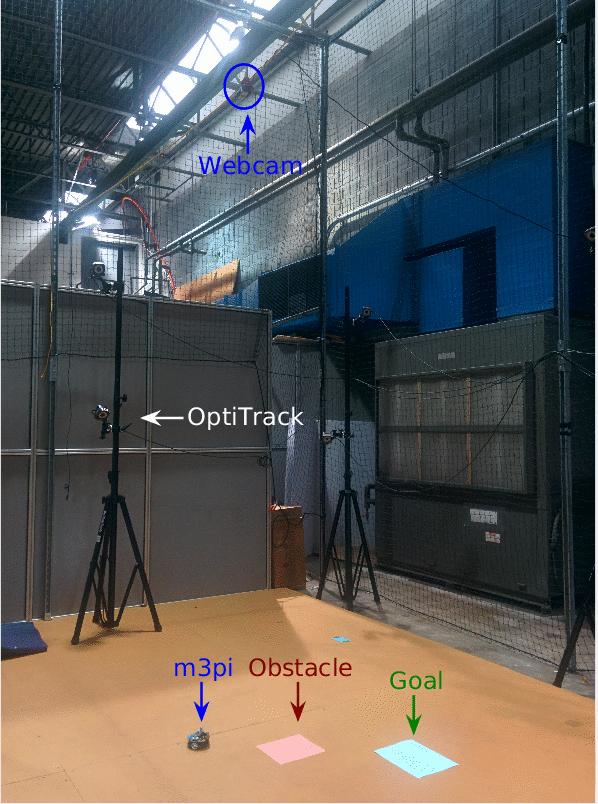}
   \caption{Experimental Setup for path planning of \textit{m3pi}
     robot}
   \label{fig:expt-setup}
\end{figure}

The workspace environment for running the experiments is shown in
Fig. \ref{fig:expt-setup}. We use an overhead webcam to obtain images
of the environment, which are then processed to determine the
locations of the predicates, represented by the papers laid on the
floor. The OptiTrack system is used to track the position of the
\textit{m3pi} robot accurately. A video of three different
experiments, as well as an user-interactive example
are uploaded \href{http://tinyurl.com/sahaResearch#MILP-MTL}{here}
\footnote{http://tinyurl.com/sahaResearch\#MILP-MTL}. Here, we present
a simulated version of one of the practical experiments in
Fig. \ref{fig:react}. The desired task specification is again a
reach-avoid criteria
\begin{align*}
  \phi_{3} = \left( \globally \neg p_{\Unsafe_{1}} \right) \wedge \left(
  \globally_{[17.5, 20]} ~ p_{\Goal} \right),
\end{align*}
where \(\Unsafe_{1}\)
is a previously known unsafe region in the workspace of the robot as
shown in Fig. \ref{fig:react-1}. We assume that the workspace is
changing dynamically such that as the robot is navigating through the
workspace towards the goal region, it might come across another unsafe
region all of a sudden. With this knowledge we encode our MILP
controller with one safe predicate and two unsafe predicates as
detailed in Section \ref{sec:ctrlSynthesis}. Since, we pass the
\((A,b)\)
pair defining any predicate as parameters to the MILP controller we do
not need to know the exact location of the \(\Unsafe_{2}\)
at the beginning of the path planning problem. When the system becomes
aware of the new unsafe region it updates the corresponding \((A,b)\)
pair values and solves for a new feasible trajectory from its current
position using Algorithm \ref{alg:2}. As shown in
Fig. \ref{fig:react}, a previously unknown unsafe area pops up
at \(7.5\)
seconds and as can be seen from Fig. \ref{fig:react-3} -
Fig. \ref{fig:react-5} the robot finds a feasible trajectory in \(3\)
time-steps, due to the reason that we limit the duration of steps
\(3\)
to \(8\) of the Algorithm \ref{alg:2} to time-step of \(0.5\) seconds.

% \begin{figure*}[!t]
%   \centering \subfloat[][Feasible trajectory obtained using
%   BluSTL.]{\includegraphics[scale=0.2]{./pics/react-sim-1.jpg}
%     \label{fig:react-sim-1}} \hfil
%   \subfloat[Feasible trajectory obtained
% using our
% approach.]{\includegraphics[scale=0.2]{./pics/react-sim-2}
%   \label{fig:react-sim-2}} \hfil
% \subfloat[Feasible trajectory obtained
% using our
% approach.]{\includegraphics[scale=0.2]{./pics/react-sim-3}
%   \label{fig:react-sim-2}}
% \caption{Path planning solutions of the mobile robot with complex task
% specifications. The workspace is same as before.}
% \label{fig:react-sim}
% \end{figure*}

\begin{figure*}[tpbh]
  \vspace{-0.5in}
\centering
\begin{tabular}{ccc}
  \subfloat[\label{fig:react-1} Known environment configuration.]{\includegraphics[scale=0.2]{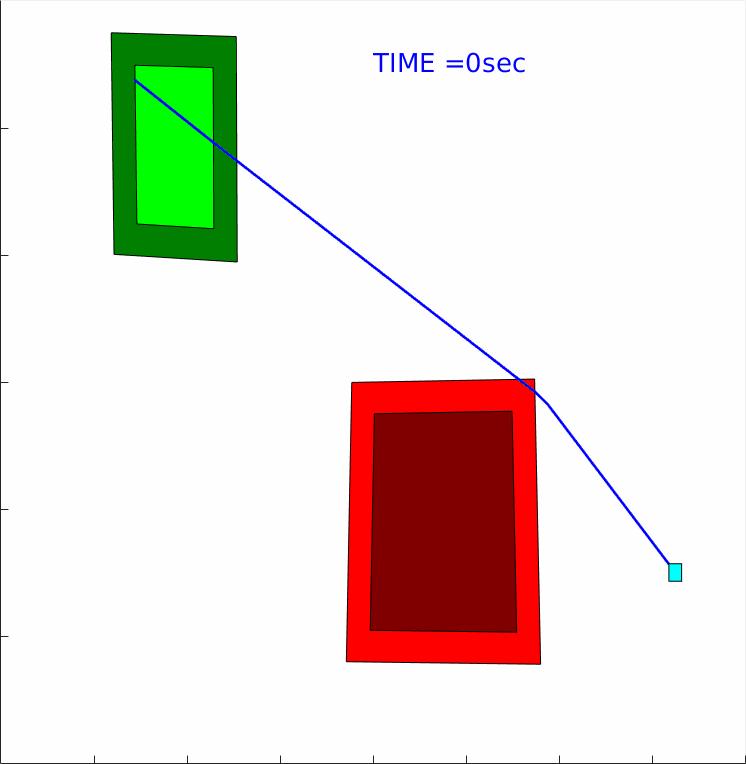}} & 
  \subfloat[\label{fig:react-2} New unsafe region
                                                                                                                     appears.]{\includegraphics[scale=0.2]{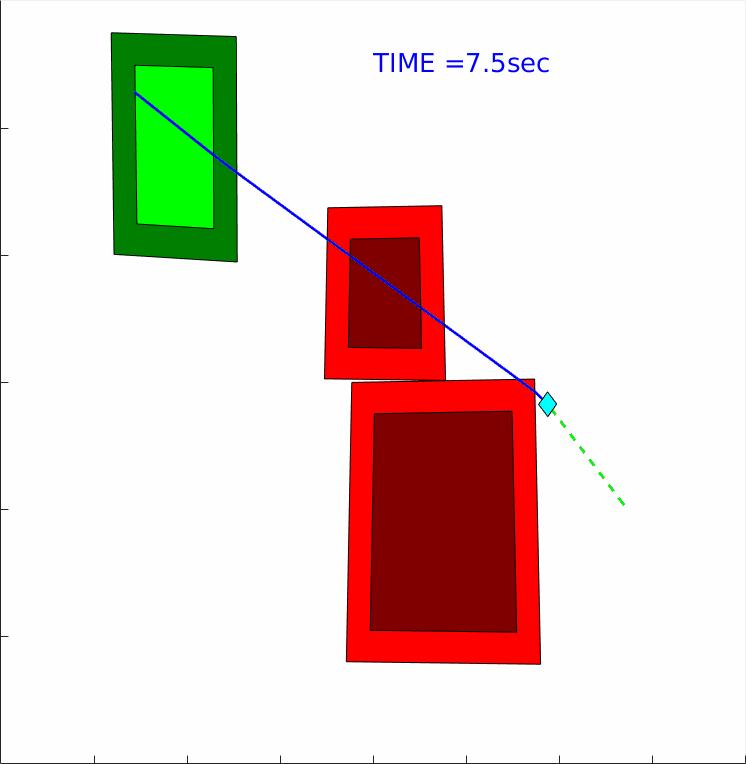}} \\ 
                                                                                                             \subfloat[\label{fig:react-3}
                                                                                                              Due
                                                                                                              to
                                                                                                              limited
                                                                                                              computational
                                                                                                              time
                                                                                                              a
                                                                                                              feasible
                                                                                                              trajectory
                                                                                                              is
                                                                                                              not
                                                                                                              found
                                                                                                              yet.]{\includegraphics[scale=0.2]{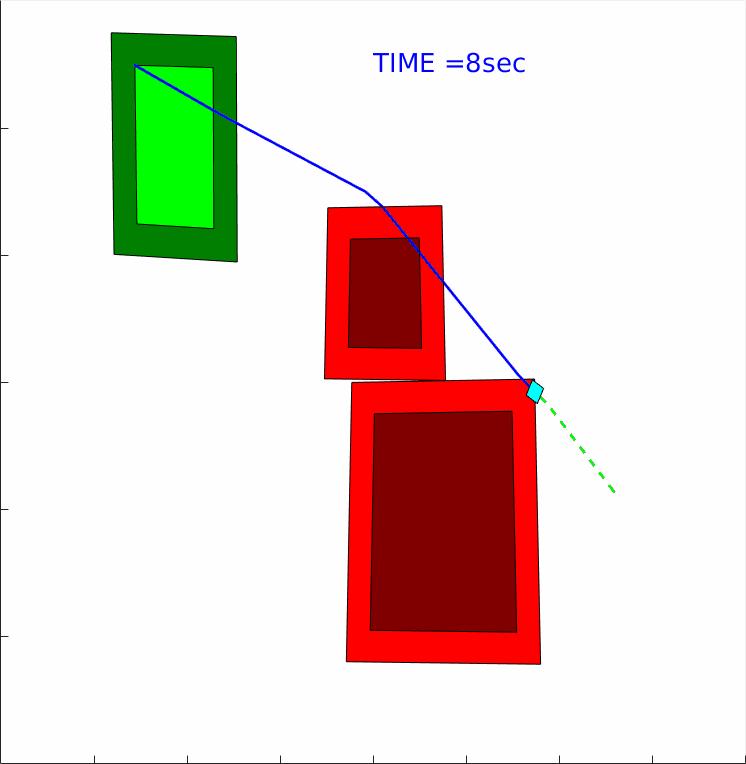}} &
  \subfloat[\label{fig:react-4} Still searching for a feasible trajectory.]{\includegraphics[scale=0.2]{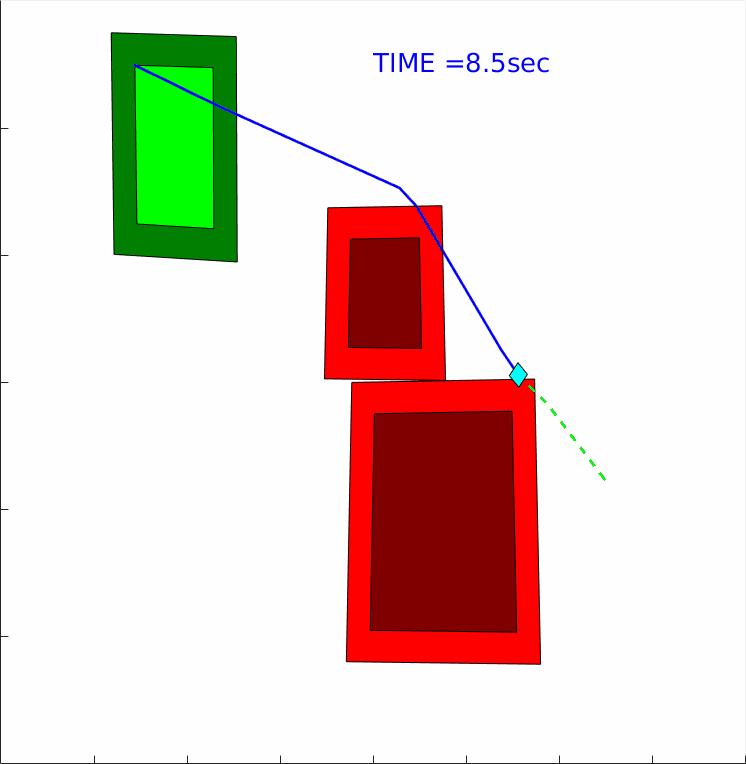}} \\
  \subfloat[\label{fig:react-5} A feasible trajectory is found at the
                                                                                                                              third
                                                                                                                              time-step
                                                                                                                              after
                                                                                                                              the
                                                                                                                              new
                                                                                                                              unsafe
                                                                                                                              region
                                                                                                                              appeared.]{\includegraphics[scale=0.2]{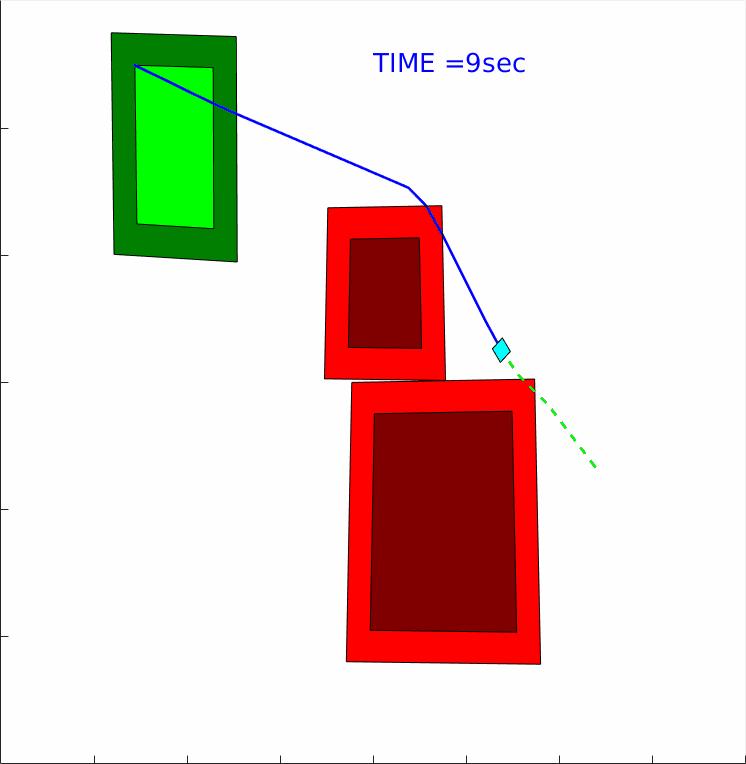}}
                                                                                                            &
  \subfloat[\label{fig:react-6} Final path taken by the robot to
                                                                                     execute
                                                                                     its
                                                                                     task specification.]{\includegraphics[scale=0.2]{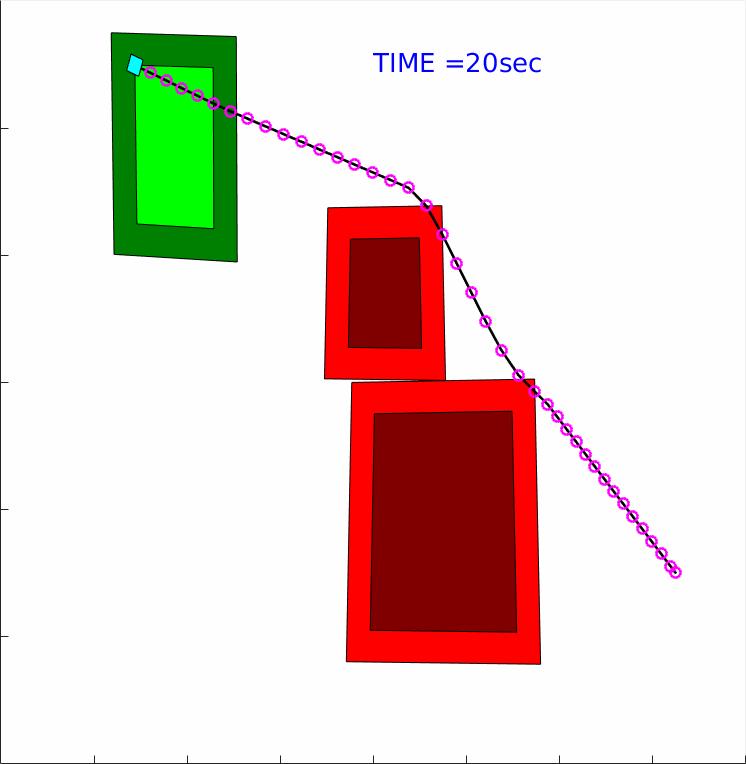}}
  \\
% \subfloat[G]{\includegraphics{logo}}&
% \subfloat[H]{\includegraphics{logo}} \\
\end{tabular}
\caption{Real-time path planning in a dynamically changing environment.}
\label{fig:react}
\vspace{-0.1in}
\end{figure*}

\section{Conclusion and Future Work}
\label{sec:conclusion}
In this paper, we consider the problem of optimizing the inputs to an
MLD system such that it satisfies an MTL specification.  We do this by
finding the trajectory points that most violate the specification,
constraining them to satisfy the corresponding predicate, and resolving
the resulting MILP optimization problem.  Although this problem can be
fully formalized as an MILP optimization problem and solved directly,
this introduces a number of binary variables and constraints that are
linear in the length of the trajectory and size of the MTL
specification. Our approach iteratively adds constraints, and require
solving MILPs multiple times rather than once but can yield a low-cost
feasible solution much faster by considering the smaller MILPs,
instead of one large MILP.

We present the efficacy of our approach by finding feasible
trajectories corresponding to two different MTL specifications by
solving the optimization problem for a mobile robot in a few
seconds. The numerical results show that this approach can generate
feasible trajectories by adding only a few of the binary variables and
constraints that would otherwise have been added. This motivates the
current use versus the full MILP, as is usually done, which would have
included hundreds of binary variables and constraints. We also show
the reactiveness of the proposed approach by implementing the
controller in real-time on a \textit{m3pi} robot in a dynamically
changing environment. Future work involves exploring heuristic
approaches so as to use a linear combination of the time-points at
which the MTL specification is violated to add the constraint, rather
than just using the critical time-points.

\section{Acknowledgments}

The research reported in this paper is partially supported by the NSF
through the grants CAREER CNS-0953976 and CNS-1218109. The authors
would like to thank Andrew Winn for some very constructive
discussions.

\bibliographystyle{IEEEtran}   % Set any options you want
\bibliography{acc16arxiv}  % Build the bibliography

% Generated by IEEEtran.bst, version: 1.13 (2008/09/30)
\begin{thebibliography}{10}
\providecommand{\url}[1]{#1}
\csname url@samestyle\endcsname
\providecommand{\newblock}{\relax}
\providecommand{\bibinfo}[2]{#2}
\providecommand{\BIBentrySTDinterwordspacing}{\spaceskip=0pt\relax}
\providecommand{\BIBentryALTinterwordstretchfactor}{4}
\providecommand{\BIBentryALTinterwordspacing}{\spaceskip=\fontdimen2\font plus
\BIBentryALTinterwordstretchfactor\fontdimen3\font minus
  \fontdimen4\font\relax}
\providecommand{\BIBforeignlanguage}[2]{{%
\expandafter\ifx\csname l@#1\endcsname\relax
\typeout{** WARNING: IEEEtran.bst: No hyphenation pattern has been}%
\typeout{** loaded for the language `#1'. Using the pattern for}%
\typeout{** the default language instead.}%
\else
\language=\csname l@#1\endcsname
\fi
#2}}
\providecommand{\BIBdecl}{\relax}
\BIBdecl

\bibitem{alur2000discrete}
R.~Alur, T.~A. Henzinger, G.~Lafferriere, and G.~J. Pappas, ``Discrete
  abstractions of hybrid systems,'' \emph{Proceedings of the IEEE}, vol.~88,
  no.~7, pp. 971--984, 2000.

\bibitem{fainekos2009temporal}
G.~E. Fainekos, A.~Girard, H.~Kress-Gazit, and G.~J. Pappas, ``Temporal logic
  motion planning for dynamic robots,'' \emph{Automatica}, vol.~45, no.~2, pp.
  343--352, 2009.

\bibitem{kloetzer2008fully}
M.~Kloetzer and C.~Belta, ``A fully automated framework for control of linear
  systems from temporal logic specifications,'' \emph{Automatic Control, IEEE
  Transactions on}, vol.~53, no.~1, pp. 287--297, 2008.

\bibitem{fainekos2005hybrid}
G.~E. Fainekos, H.~Kress-Gazit, and G.~J. Pappas, ``Hybrid controllers for path
  planning: A temporal logic approach,'' in \emph{Decision and Control, 2005
  and 2005 European Control Conference. CDC-ECC'05. 44th IEEE Conference
  on}.\hskip 1em plus 0.5em minus 0.4em\relax IEEE, 2005, pp. 4885--4890.

\bibitem{decastro2014synthesis}
J.~DeCastro and H.~Kress-Gazit, ``Synthesis of nonlinear continuous controllers
  for verifiably-correct high-level, reactive behaviors,'' \emph{International
  Journal of Robotics Research Accepted}, 2014.

\bibitem{Kressgazitcorrect}
H.~Kress-Gazit, T.~Wongpiromsarn, and U.~Topcu, ``Correct, reactive, high-level
  robot control,'' \emph{Robotics Automation Magazine, IEEE}, vol.~18, no.~3,
  pp. 65--74, Sept 2011.

\bibitem{Karaman2012sampling}
S.~Karaman and E.~Frazzoli, ``Sampling-based algorithms for optimal motion
  planning with deterministic $\mu$-calculus specifications,'' in
  \emph{American Control Conference (ACC), 2012}.\hskip 1em plus 0.5em minus
  0.4em\relax IEEE, 2012, pp. 735--742.

\bibitem{livingston2015cross}
S.~C. Livingston, E.~M. Wolff, and R.~M. Murray, ``Cross-entropy temporal logic
  motion planning,'' in \emph{Proceedings of the 18th International Conference
  on Hybrid Systems: Computation and Control}.\hskip 1em plus 0.5em minus
  0.4em\relax ACM, 2015, pp. 269--278.

\bibitem{karaman2011linear}
S.~Karaman and E.~Frazzoli, ``Linear temporal logic vehicle routing with
  applications to multi-uav mission planning,'' \emph{International Journal of
  Robust and Nonlinear Control}, vol.~21, no.~12, pp. 1372--1395, 2011.

\bibitem{wolff2014optimization}
E.~M. Wolff, U.~Topcu, and R.~M. Murray, ``Optimization-based trajectory
  generation with linear temporal logic specifications,'' in \emph{Robotics and
  Automation (ICRA), 2014 IEEE International Conference on}.\hskip 1em plus
  0.5em minus 0.4em\relax IEEE, 2014, pp. 5319--5325.

\bibitem{karaman2008optimal}
S.~Karaman, R.~G. Sanfelice, and E.~Frazzoli, ``Optimal control of mixed
  logical dynamical systems with linear temporal logic specifications,'' in
  \emph{Decision and Control, 2008. CDC 2008. 47th IEEE Conference on}.\hskip
  1em plus 0.5em minus 0.4em\relax IEEE, 2008, pp. 2117--2122.

\bibitem{karaman2008vehicle}
S.~Karaman and E.~Frazzoli, ``Vehicle routing problem with metric temporal
  logic specifications,'' in \emph{Decision and Control, 2008. CDC 2008. 47th
  IEEE Conference on}.\hskip 1em plus 0.5em minus 0.4em\relax IEEE, 2008, pp.
  3953--3958.

\bibitem{raman2014model}
V.~Raman, A.~Donz{\'e}, M.~Maasoumy, R.~M. Murray, A.~Sangiovanni-Vincentelli,
  S.~Seshia \emph{et~al.}, ``Model predictive control with signal temporal
  logic specifications,'' in \emph{Decision and Control (CDC), 2014 IEEE 53rd
  Annual Conference on}.\hskip 1em plus 0.5em minus 0.4em\relax IEEE, 2014, pp.
  81--87.

\bibitem{wongpiromsarn2012receding}
T.~Wongpiromsarn, U.~Topcu, and R.~M. Murray, ``Receding horizon temporal logic
  planning,'' \emph{Automatic Control, IEEE Transactions on}, vol.~57, no.~11,
  pp. 2817--2830, 2012.

\bibitem{ulusoy2013receding}
A.~Ulusoy, M.~Marrazzo, and C.~Belta, ``Receding horizon control in dynamic
  environments from temporal logic specifications.'' in \emph{Robotics: Science
  and Systems}, 2013.

\bibitem{raman2015reactive}
V.~Raman, A.~Donz{\'e}, D.~Sadigh, R.~M. Murray, and S.~A. Seshia, ``Reactive
  synthesis from signal temporal logic specifications,'' in \emph{Proceedings
  of the 18th International Conference on Hybrid Systems: Computation and
  Control}.\hskip 1em plus 0.5em minus 0.4em\relax ACM, 2015, pp. 239--248.

\bibitem{fainekos2009robustness}
G.~E. Fainekos and G.~J. Pappas, ``Robustness of temporal logic specifications
  for continuous-time signals,'' \emph{Theoretical Computer Science}, vol. 410,
  no.~42, pp. 4262--4291, 2009.

\bibitem{koymans1990specifying}
R.~Koymans, ``Specifying real-time properties with metric temporal logic,''
  \emph{Real-time systems}, vol.~2, no.~4, pp. 255--299, 1990.

\bibitem{fainekos2006robustness}
G.~E. Fainekos and G.~J. Pappas, \emph{Robustness of temporal logic
  specifications}.\hskip 1em plus 0.5em minus 0.4em\relax Springer, 2006.

\bibitem{garey1979computers}
M.~R. Garey and D.~S. Johnson, ``Computers and intractability: a guide to the
  theory of {NP}-completeness. 1979,'' \emph{San Francisco, LA: Freeman}, 1979.

\bibitem{earl2005iterative}
M.~G. Earl and R.~D'Andrea, ``Iterative milp methods for vehicle-control
  problems,'' \emph{Robotics, IEEE Transactions on}, vol.~21, no.~6, pp.
  1158--1167, 2005.

\bibitem{Bemporad99}
A.~Bemporad and M.~Morari, ``Control of systems integrating logic, dynamics and
  constraints,'' \emph{Automatica}, vol.~35, no.~3, pp. 407--427, 1999.

\bibitem{Khalil}
H.~K. Khalil, \emph{Nonlinear Systems}, 3rd~ed.\hskip 1em plus 0.5em minus
  0.4em\relax Prentice Hall, 2002.

\bibitem{Julius12}
A.~A. Julius and A.~K. Winn, ``Safety controller synthesis using human
  generated trajectories: Nonlinear dynamics with feedback linearization and
  differential flatness,'' in \emph{Proc. American Control Conference},
  Montreal, Canada., 2012, pp. 709--714.

\bibitem{annpureddy2011s}
Y.~Annpureddy, C.~Liu, G.~Fainekos, and S.~Sankaranarayanan, \emph{S-taliro: A
  tool for temporal logic falsification for hybrid systems}.\hskip 1em plus
  0.5em minus 0.4em\relax Springer, 2011.

\bibitem{lofberg2004yalmip}
\BIBentryALTinterwordspacing
J.~L{\"o}fberg, ``{YALMIP}: A toolbox for modeling and optimization in
  {MATLAB},'' in \emph{Computer Aided Control Systems Design, 2004 IEEE
  International Symposium on}.\hskip 1em plus 0.5em minus 0.4em\relax IEEE,
  2004, pp. 284--289. [Online]. Available:
  \url{http://users.isy.liu.se/johanl/yalmip}
\BIBentrySTDinterwordspacing

\end{thebibliography}

\end{document}